\documentclass[prl,aps,twocolumn,superscriptaddress,dvipsnames,amssymb,9pt]{revtex4}
\usepackage{graphicx}
\usepackage{amsmath,amssymb,amsthm,amsfonts,mathptmx}
\usepackage{color}
\usepackage[matrix,frame,arrow]{xypic}
\usepackage{times}
\usepackage{algorithm}
\floatname{algorithm}{Protocol}

\theoremstyle{plain}

\newtheorem{theorem}{Theorem}
\newtheorem{corollary}{Corollary}
\newtheorem{lemma}{Lemma}
\theoremstyle{definition}

\newcommand{\bra}[1]{\langle#1|}
\newcommand{\ket}[1]{|#1\rangle}
\newcommand{\braket}[2]{\langle#1|#2\rangle}
\newcommand{\ketbra}[2]{|#1\rangle\!\langle#2|}
\DeclareMathOperator{\Tr}{Tr}

\begin{document}
\title{Limitations on information theoretically secure quantum homomorphic encryption}
\author{Li Yu}
\affiliation{Singapore University of Technology and Design, 20 Dover Drive, Singapore 138682}
\affiliation{Centre for Quantum Technologies, National University of Singapore, 3 Science Drive 2, Singapore 117543}
\author{Carlos A. P\'erez-Delgado}
\affiliation{Singapore University of Technology and Design, 20 Dover Drive, Singapore 138682}
\author{Joseph F. Fitzsimons}\email{joe.fitzsimons@nus.edu.sg}
\affiliation{Singapore University of Technology and Design, 20 Dover Drive, Singapore 138682}
\affiliation{Centre for Quantum Technologies, National University of Singapore, 3 Science Drive 2, Singapore 117543}
\begin{abstract}
Homomorphic encryption is a form of encryption which allows computation to be carried out on the encrypted data without the need for decryption. The success of quantum approaches to related tasks in a delegated computation setting has raised the question of whether quantum mechanics may be used to achieve information theoretically secure fully homomorphic encryption. Here we show, via an information localisation argument, that deterministic fully homomorphic encryption necessarily incurs exponential overhead if perfect security is required.
\end{abstract}
\maketitle

The insight that information must be represented and manipulated in accordance with physical laws has led to the blossoming field of quantum information science. The applications of this approach to information processing are diverse, and it has led to discoveries ranging from new algorithms \cite{shor1994algorithms,grover1996fast} and communications protocols \cite{bennett1992communication, bennett1993teleporting} which exploit quantum states for increased efficiency to techniques for enhancing the precision of metrology \cite{giovannetti2004quantum}. Historically, cryptography was one of the first fields for which quantum information processing was shown to offer an advantage over classical processing, when in 1984 Bennett and Brassard introduced a quantum protocol for information theoretically secure key distribution \cite{bb84}. While for many quantum cryptography has remained synonymous with quantum key distribution, the field has grown substantially, with quantum protocols being discovered which enhance the security with which many cryptographic tasks can be accomplished, including digital signatures \cite{gottesman2001quantum}, anonymous communication \cite{brassard2007anonymous}, private database queries \cite{giovannetti2008quantum}, and random number generation \cite{pironio2010random}. The importance of information theoretically secure cryptography is highlighted by the fact that quantum algorithms offer new attacks against cryptosystems which rely on assumptions of computational intractability for their security \cite{yan2007cryptanalytic,boneh1995quantum,brassard1998quantum}. Unfortunately, not all cryptographic tasks that we may wish to accomplish admit an information theoretically secure quantum protocol, and indeed a number of no-go theorems have been discovered which show that quantum mechanics alone is insufficient to accomplish certain tasks, such as bit commitment \cite{mayers1997unconditionally} and oblivious transfer \cite{lo1997insecurity}, with perfect security.

One of the most celebrated results in classical cryptography in recent years has been the discovery of computationally secure protocols for fully homomorphic computation \cite{Gentry09,Dijk09,smart2010fully,brakerski2011efficient}. A homomorphic encryption scheme is one which allows data to be encrypted in such a way that certain operations can be performed on the data without decryption. This allows a user to provide encrypted data to a remote server for processing without having to reveal the plaintext. A number of examples of such homomorphic encryption schemes have been known for many years \cite{rivest1978data}, but it was the ground-breaking work of Gentry \cite{Gentry09} which for the first time demonstrated a fully homomorphic encryption scheme, one which allowed for arbitrary computations to be performed on the encrypted data, rather than being restricted to some class of non-universal operations. The ability to perform universal computation on encrypted data has greatly increased the utility of homomorphic encryption, and as a result it has become one of the most active areas of modern cryptography.

One draw back of known fully homomorphic schemes is that they derive their security from computational assumptions. The existence of perfectly secure quantum protocols for blind computation \cite{bfk09,fitzsimons2012unconditionally,morimae2013blind,mantri2013optimal}, and recent experimental demonstrations thereof \cite{barz2012demonstration,barz2013experimental}, highlight the possibilities opened by quantum cryptographic techniques in this area. As cryptographic tasks, blind computation and homomorphic encryption are similar in many ways. Both tasks envision a two party scenario, where the first party, Alice, wishes the second party, Bob, to carry out a computation for her, without revealing the input of her computation. In blind computation, however, Alice specifies not only the input data but also the computation to be performed, and the task is to utilise Bob's resources to perform this computation without revealing either the input or the program. As a result, the current protocols for accomplishing this task are interactive, requiring multiple rounds of communication between Alice and Bob, a significant difference from the setting of homomorphic encryption.

The idea of quantum homomorphic encryption appears in \cite{MinL13}, which shows that a perfect, universal, quantum homomorphic scheme cannot be constructed using one-time pads and which presents an interactive protocol for achieving similar functionality. Other cryptographic schemes have been proposed that achieve some of the functionality of homomorphic encryption using quantum data \cite{Ch05,Fisher13}. However, these rely on assisted computation, and so require multiple rounds of interaction between Alice and Bob, thus amounting to interactive protocols rather than simply encryption schemes. A quantum homomorphic encryption scheme does exist for a restricted model of quantum computation known as boson scattering, which offers limited information theoretic security \cite{rfg12}. The existence of such schemes raises the question as to whether quantum techniques can be exploited to construct an information theoretically secure fully homomorphic encryption scheme. Here we answer that question in the negative by proving that quantum mechanics does not allow for efficient information theoretically secure fully homomorphic encryption that perfectly conceals the plaintext. To achieve this we first formalise the notion of quantum homomorphic encryption, and then proceed to show via an information localisation argument that any such scheme which perfectly hides Alice's input must necessarily reveal the computation performed, and hence the encoding must be sufficiently long to specify any such computation. For a fully homomorphic encryption scheme this implies that the coding must be exponentially long, and thus rules out the existence of efficient fully homomorphic encryption schemes which perfectly hide Alice's data.

Formally, a classical homomorphic encryption scheme consists of four procedures. The first is a key generation algorithm that generates a classical encryption key, a classical decryption key, and potentially some additional auxiliary key. The second is an encryption algorithm, that encrypts the input using the encryption key. Third is a decryption algorithm that decrypts the output using the decryption key. Finally, there is an evaluation algorithm that performs the computation on the ciphertext without decryption, which may use the auxiliary key. For any permissible logical circuit $C$, the result of the evaluation algorithm should be such that after decrypting the output, one obtains the result of applying $C$ to the unencrypted input. A fully homomorphic encryption scheme, then, is one in which $C$ can be freely chosen from the set of all classical circuits. Here we shall consider only schemes with perfect completeness, where the evaluation operator must deterministically implement the chosen circuit. We will say that a homomorphic encryption scheme has {\em perfect information theoretic security} if the mutual information between the plaintext and ciphertext is zero. 

\begin{figure}[t!]
\centering
\includegraphics[width=\columnwidth]{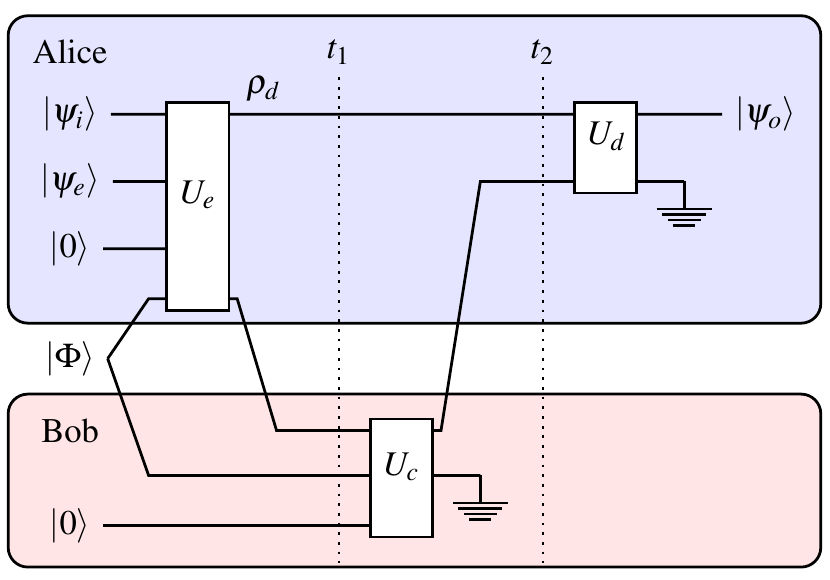}
\caption{A schematic diagram for a general quantum homomorphic encryption scheme with input data $\ket{\psi_i}$ and output $\ket{\psi_o}$. The state $\ket{\psi_e}$ represents the initial state of Alice's key, while $U_e$ and $U_d$ are Alice's encryption and decryption operators. Both parties are also allowed an ancilla system, and access to a shared entanglement resource. Alice's decryption key corresponds to the subsystem she retains after applying $U_e$ to her system. Note that no assumption is made about the dimensionality of subsystems. Time points $t_1$ and $t_2$, used in the proof of Theorem \ref{thm:no-homo}, are also shown.
}\label{fig:QHE}
\end{figure}

We will define a {\em quantum homomorphic encryption} (QHE) scheme using similar criteria as for the classical case, extended to take into account the possibility of entanglement within the protocol. A QHE scheme consists of four components: a key generation protocol which produces a quantum state $\ket{\psi_e}$ used as a key for encryption; an encryption unitary operator $U_e$ which encrypts the input state $\ket{\psi_i}$ using the encryption key state, potentially making use of some ancilla system, and which produces a decryption key in a state $\rho_d$; a decryption unitary operator $U_d$ which decrypts the encrypted state using the key state; and a set of evaluation unitary operators $\{U_C\}$, such that after decrypting the output the net effect is equivalent to applying the quantum circuit $C$ directly to the initial input state. Here the decryption key is produced when the encryption unitary is applied. Although this is somewhat more general than the procedure for generating the corresponding classical key, we make this generalization to allow for the possibility of a causal relationship between encryption and decryption keys which, via the no-cloning theorem, may prevent them from existing simultaneously. Note that we have not specified an auxiliary key. This is because, without loss of generality, we can assume that this auxiliary key forms part of the encrypted state. An encryption-evaluation-decryption sequence based on this definition is depicted in Figure \ref{fig:QHE}.

As we now prove, for such a scheme to operate deterministically, it is necessary that the dimension of the encrypted state grows as the log of the cardinality of the set of possible choices of $C$, and hence fully homomorphic encryption with perfect information theoretic security is impossible except when the size of the encoding grows exponentially with the size of the plaintext. To prove this, we begin by proving a modified version of an information localisation theorem due to Griffiths \cite{Grif05}.

\begin{lemma}[Data Localisation]\label{lemma:DL}
Let $S$ be some bipartite quantum system with Hilbert space $\mathcal{H}_A \otimes \mathcal{H}_B$, initially in state $\left(\ket{\psi} \otimes \ket{\phi}\right)_A \otimes \ket{\gamma}_B$, where $\ket{\phi}$ and $\ket{\gamma}$ are fixed states. Let $\rho$ be the state of $S$ after the application of a unitary operator $U$.  Then, if the mutual information $I(\Tr_{A}\rho ; \ket{\psi}\bra{\psi}) = 0$, there exists a unitary operator $V: \mathcal{H}_A \mapsto \mathcal{H}_A$ such that $\Tr_{B}\rho = V \left(\ket{\psi}\bra{\psi} \otimes \sigma\right) V^\dag$ for some density matrix $\sigma$ independent of $\ket{\psi}$.
\end{lemma}
\begin{proof}
For simplicity of notation we will define $\rho_A = \Tr_{B}\rho$ and $\rho_B = \Tr_{A}\rho$, and use $r$ to denote the rank of $\rho_B$. We shall further divide the Hilbert space $\mathcal{H}_A = \mathcal{H}_{A_1} \otimes \mathcal{H}_{A_2}$ such that $\ket{\psi}$ is the state of $\mathcal{H}_{A_1}$ and $\ket{\phi}$ is the state of $\mathcal{H}_{A_2}$.

We begin from the requirement that $I(\Tr_{A}\rho ; \ket{\psi}\bra{\psi}) = 0$. This implies that changing the value of $\ket{\psi}$, while holding $\ket{\phi}$ and $\ket{\gamma}$ constant, will not alter $\Tr_{A}\rho$. We shall consider the effect on the state of $\rho$ of varying only $\ket{\psi}$. Let an orthonormal basis of $\mathcal{H}_{A_1}$ as $\ket{a^j}$, $j=1,2,\cdots,d_{A_1}$, where $d_{A_1}$ is the dimension of $\mathcal{H}_{A_1}$, and let $\{\ket{b^k}: k=1,2,\cdots,r\}$ be an orthonormal set of eigenstates of $\rho_B$ with corresponding eigenvalues $p_k$. 

For each state $\ket{a^j}$, for $1\leq j\leq d_{A_1}$, we can expand the state of the system after application of $U$ to yield
\begin{multline}\label{eq1}
U (\ket{a^j}_{A_1} \otimes \ket{\phi}_{A_2}\otimes \ket{\gamma}_B) = \sum_{k=1}^{r} \sqrt{p_k} \ket{\tau^{jk}}_{A} \otimes \ket{b^k}_B.
\end{multline}
Note that the possible complex phases have been absorbed into the definition of $\ket{\tau^{jk}}$. Since $\ket{b^k}$ are eigenstates of $\rho_B$ with eigenvalues $p_k$, the expansion on the right hand side of Eq.~\eqref{eq1} is a Schmidt expansion for $U (\ket{a^j}_{A_1} \otimes \ket{\phi}_{A_2}\otimes \ket{\gamma}_B)$, and hence $\{\ket{\tau^{jk}}: k=1,2,\cdots,r\}$ for any fixed $j$ must be orthonormal.  Thus, we have  $\braket{\tau^{jk}}{\tau^{j k'}}=\delta_{k,k'}$.

Now consider the case where we keep the input on $\mathcal{H}_{A_2}$ and $\mathcal{H}_{B}$ fixed, while changing the input state on $\mathcal{H}_{A_1}$ to one of the form $\ket{\upsilon^{j j'}}=(\ket{a^j}+ \ket{a^{j'}})/\sqrt{2}$ for $j\neq j'$. In this case,
	\begin{multline}\label{eq3}
	U (\ket{\upsilon^{j j'}}_{A_1} \otimes \ket{\phi}_{A_2}\otimes \ket{\gamma}_B) = \\
	\sum_{k=1}^{r} \sqrt{p_k} \left[(\ket{\tau^{jk}} + \ket{\tau^{j' k}})/\sqrt{2}\right]_{A} \otimes \ket{b^k}_B.
	\end{multline}
	Since the output reduced density operator on $\mathcal{H}_B$ is still $\rho_B=\sum_{k=1}^r p_k \ketbra{b^k}{b^k}$, the right hand side of Eq.~\eqref{eq3} should be a Schmidt expansion, with the Schmidt coefficients still being $\sqrt{p_k}$. Hence $(\ket{\tau^{jk}} + \ket{\tau^{j' k}})/\sqrt{2}$ must be already normalised and these states must be orthogonal for different values of $k$. From this we obtain
	\begin{align}\label{eq6}
	\delta_{k,k'} &= \frac{1}{2} (\bra{\tau^{jk}} + 
	\bra{\tau^{j' k}}) (\ket{\tau^{jk'}} + \ket{\tau^{j' k'}})  \notag\\
	&= \delta_{k,k'}+ \frac{1}{2} (\braket{\tau^{jk}}{\tau^{j' k'}} + \braket{\tau^{j' k}}{\tau^{j k'}}), \quad j\neq j',
	\end{align}
and hence $\braket{\tau^{jk}}{\tau^{j' k'}} + \braket{\tau^{j' k}}{\tau^{j k'}}=0$ as long as $j\neq j'$.

	Similarly, by considering input states on $\mathcal{H}_{A_1}$ of the form $\ket{\eta^{j j'}}=(\ket{a^j} + i \ket{a^{j'}})/\sqrt{2}$, we obtain $\braket{\tau^{jk}}{\tau^{j' k'}} - \braket{\tau^{j' k}}{\tau^{j k'}}=0$ and hence $\braket{\tau^{jk}}{\tau^{j' k'}} =0$ for $j\neq j'$. These criteria can be expressed compactly as $\braket{\tau^{jk}}{\tau^{j' k'}}=\delta_{j,j'} \delta_{k,k'}$. Hence $\{\ket{\tau^{jk}}\}$ forms an orthonormal set, and it is possible to \emph{define} the subspaces $\mathcal{H}_C$ and $\mathcal{H}_D$ as having orthonormal bases $\{\ket{c^j}\}$ and $\{\ket{d^k}\}$ such that $\mathcal{H}_{A} = \mathcal{H}_C \otimes \mathcal{H}_D$, and
	\begin{equation}\label{eq15}
	\ket{\tau^{jk}}=\ket{c^j} \otimes \ket{d^k}, \quad j=1,2,\cdots,d_{A_1}, \quad k=1,2,\cdots,r.
	\end{equation}
For a generic input state $\ket{\xi} = \ket{\psi} \otimes \ket{\phi} \otimes \ket{\gamma}$, where $\ket{\psi}=\sum_{j=1}^{d_{A_1}} \alpha_j \ket{a^j}$ we then have
	\begin{align}\label{eq16}
	U \ket{\xi}  &= \sum_{j=1}^{d_{A_1}} \alpha_j \sum_{k=1}^r \sqrt{p_k} \ket{\tau^{jk}}_{A} \otimes \ket{b^k}_B \notag\\
	&=\left(\sum_{j=1}^{d_{A_1}} \alpha_j \ket{c^j} \right)_C \otimes \left(\sum_{k=1}^r \sqrt{p_k} \ket{d^k} \otimes \ket{b^k} \right)_{DB} \notag\\
	&=\ket{\psi}_C \otimes \left(\sum_{k=1}^r \sqrt{p_k} \ket{d^k} \otimes \ket{b^k} \right)_{DB}.
	\end{align}	
Now, let $V' : \mathcal{H}_{A_1} \mapsto \mathcal{H}_{C}$ be an isometry such that
\begin{equation}
V' \ket{a^j} = \ket{c^j},  \quad j=1,2,\cdots,d_{A_1},
\end{equation}
and let $V$ be any extension of $V'$ into a full unitary over $\mathcal{H}_A$. Then $\Tr_B \left(U \ket{\xi}\bra{\xi} U^\dagger \right) =  V \left(\ket{\psi}\bra{\psi} \otimes \sigma\right)V^\dag$, for some density operator $\sigma$ independent of $\ket{\psi}$, as the lemma requires.
\end{proof}

Lemma \ref{lemma:DL} shows that in any quantum homomorphic encryption scheme with perfect information theoretic security, the computation has to occur on Alice's ``side''. The following theorem formalises this intuition, showing that the encrypted state must contain enough information to identify any operator previously applied to it.

\begin{theorem}\label{thm:no-homo}
Let $Q$ be a quantum homomorphic encryption scheme with perfect information theoretic security with encryption operator $U_e$ and decryption operator $U_d$ and a set of evaluation unitaries. Let $\rho_b$ ($\rho_b'$) be the state of the encrypted system after application of evaluation unitary $U_c$ ($U_{c'}$) corresponding to a quantum circuit $c$ ($c'$), to an input state $\ket{\psi_i}$. Then, if $b$ and $b'$ implement distinct unitary operations, $\rho_b$ and $\rho_b'$ must have orthogonal support.
\end{theorem}
\begin{proof}
For clarity, we will identify different parts of the encryption, circuit evaluation and decryption process with two parties, Alice and Bob, as depicted in Figure \ref{fig:QHE}. We begin by analysing the state of Alice and Bob's joint system after Alice has sent her encoded data to Bob. This is marked as time $t_1$ in Figure \ref{fig:QHE}. Let $\rho_{a,1}$ ($\rho_{b,1}$) be the states of Alice's (Bob's) subsystem at this point. From this point forward, all communication flows from Bob to Alice. The requirement that $Q$ be perfectly information theoretically secure implies that $I( \rho_{b,1} ; \ketbra{\psi_i}{\psi_i}) = 0$.
Hence, by Lemma \ref{lemma:DL} there exists some unitary operator $V$ such that
\begin{equation}
\rho_{a,1} = V (\ketbra{\psi_i}{\psi_i} \otimes \rho_{a,1}') V^\dag,	
\end{equation}
for some appropriate $\rho_{a,1}'$.

Now, consider the system after Bob has sent his message back to Alice. This is time $t_2$ in Figure \ref{fig:QHE}. Due to the previous analysis the state of the system at this point can be written as
\begin{equation}
\rho_{a,2}  = \left( V \otimes I \right)\ketbra{\psi_i}{\psi_i} \otimes \rho_{a,1}' \otimes \rho_{b} \left( V^\dag \otimes I \right),
\end{equation}
where $\rho_{b}$ represents Bob's message. The density matrix $\rho_b$ cannot in general be assumed to be pure, since Bob could have sent a message that remains entangled to his system. Here $V$ acts only on the part of the system that was in Alice's possession prior to receiving the message from Bob, and the identity operator $I$ acts on Bob's message.

The requirement that the evaluation unitary $U_c$ implements a specific circuit $c$ implies that
\begin{align}
U_d \rho_{a,2} {U_d}^\dag = \left( W_c \ketbra{\psi_i}{\psi_i} W_c^\dag \right)  \otimes \rho_\text{anc},
\end{align}
where $W_c$ is the unitary operator corresponding to quantum circuit $c$, and $\rho_\text{anc}$ is simply some state of the ancilla system. Let $U_d' = U_d \left( V \otimes I \right)$, then for all $c$ and all $\ket{\psi_i}$,
\begin{equation}
U_d' \left( \ketbra{\psi_i}{\psi_i} \otimes \rho_{a,1}' \otimes \rho_b \right) U_d'^\dag =  \left(W_c \ketbra{\psi_i}{\psi_i} {W_c}^\dag \right) \otimes \rho_\text{anc}.
\end{equation}
As the state $\rho_{a,1}'$ and the operator $U_d'$ are independent of $c$, in the language of \cite{nc97} this corresponds to a programmable quantum gate array, where $\rho_b$ acts as a \emph{program} to implement the unitary operator $W_c$. The \emph{no programming} theorem \cite{nc97} states that for a programmable quantum gate array to implement two distinct unitary operators, the program states must be orthogonal. Hence if $\rho_b$ and $\rho_{b'}$ correspond to the messages returned from Bob after application of evaluation operators corresponding to two non-equivalent circuits, then $\rho_b$ and $\rho_{b'}$ must have orthogonal support.
\end{proof}

A direct consequence of this theorem is that for any perfectly information theoretically secure homomorphic scheme (fully homomorphic or otherwise), if a known input state is encrypted, and an evaluation operator from some unknown circuit $c$ is applied, it is always possible to unambiguously determine $c$ from the resulting encrypted state. This mirrors a result obtained for one time programs \cite{broadbent2013quantum}, a similar task in which the secret to be protected is Bob's circuit rather than Alice's input. Further, this property severely compromises the efficiency of any QHE encoding, as we now prove.

\begin{corollary}
Let $Q$ be a QHE scheme, with perfect information theoretic security, that corresponds to a permissible set of operations $S$. Then the size of the system required to store the encrypted state after the application of an evaluation operator $U_c$ corresponding to an arbitrary operation in $S$ is at least $\log_2|S|$ qubits. Further, if $S$ contains the set of reversible classical operations on $n$ bits, then the size of the encrypted state grows exponentially in $n$. 
\end{corollary}
\begin{proof}
The proof of the first part of the corollary follows directly from Theorem \ref{thm:no-homo}. Each $\rho_b$ corresponding to an operator in $S$ must have orthogonal support on a distinct subspace. Since each such density operator must have at least unit rank, a system must be at least $|S|$-dimensional in order to represent every possible $\rho_b$. The final part of the corollary follows from the fact that there are $(2^n)!$ distinct permutations of the $n$-bit classical states, and hence any $S$ which contains all such operations must have cardinality at least $\log_2 (2^n)! \geq 2^n$.
\end{proof}

From this corollary, it follows that no QHE with perfect information theoretic security can deterministically implement either universal quantum computation or reversible classical computation without incurring exponential overhead, and hence in order to obtain an information theoretically secure QHE, one must be willing to sacrifice either perfect information security, determinism, or face restriction to a permissible set of circuits which is polynomial in the size of the input.

The authors thank Joshua Kettlewell, Yingkai Ouyang and Si-Hui Tan for helpful discussions. The authors acknowledge support from Singapore's National Research Foundation and Ministry of Education. This material is based on research funded by the Singapore National Research Foundation under NRF Award NRF-NRFF2013-01.

\bibliographystyle{apsrev}
\bibliography{homomorphic}

\begin{thebibliography}{33}
\expandafter\ifx\csname natexlab\endcsname\relax\def\natexlab#1{#1}\fi
\expandafter\ifx\csname bibnamefont\endcsname\relax
  \def\bibnamefont#1{#1}\fi
\expandafter\ifx\csname bibfnamefont\endcsname\relax
  \def\bibfnamefont#1{#1}\fi
\expandafter\ifx\csname citenamefont\endcsname\relax
  \def\citenamefont#1{#1}\fi
\expandafter\ifx\csname url\endcsname\relax
  \def\url#1{\texttt{#1}}\fi
\expandafter\ifx\csname urlprefix\endcsname\relax\def\urlprefix{URL }\fi
\providecommand{\bibinfo}[2]{#2}
\providecommand{\eprint}[2][]{\url{#2}}

\bibitem[{\citenamefont{Shor}(1994)}]{shor1994algorithms}
\bibinfo{author}{\bibfnamefont{P.~W.} \bibnamefont{Shor}}, in
  \emph{\bibinfo{booktitle}{Foundations of Computer Science, 1994 Proceedings.,
  35th Annual Symposium on}} (\bibinfo{organization}{IEEE},
  \bibinfo{year}{1994}), pp. \bibinfo{pages}{124--134}.

\bibitem[{\citenamefont{Grover}(1996)}]{grover1996fast}
\bibinfo{author}{\bibfnamefont{L.~K.} \bibnamefont{Grover}}, in
  \emph{\bibinfo{booktitle}{Proceedings of the twenty-eighth annual ACM
  symposium on Theory of computing}} (\bibinfo{organization}{ACM},
  \bibinfo{year}{1996}), pp. \bibinfo{pages}{212--219}.

\bibitem[{\citenamefont{Bennett and Wiesner}(1992)}]{bennett1992communication}
\bibinfo{author}{\bibfnamefont{C.~H.} \bibnamefont{Bennett}} \bibnamefont{and}
  \bibinfo{author}{\bibfnamefont{S.~J.} \bibnamefont{Wiesner}},
  \bibinfo{journal}{Physical review letters} \textbf{\bibinfo{volume}{69}},
  \bibinfo{pages}{2881} (\bibinfo{year}{1992}).

\bibitem[{\citenamefont{Bennett et~al.}(1993)\citenamefont{Bennett, Brassard,
  Cr{\'e}peau, Jozsa, Peres, and Wootters}}]{bennett1993teleporting}
\bibinfo{author}{\bibfnamefont{C.~H.} \bibnamefont{Bennett}},
  \bibinfo{author}{\bibfnamefont{G.}~\bibnamefont{Brassard}},
  \bibinfo{author}{\bibfnamefont{C.}~\bibnamefont{Cr{\'e}peau}},
  \bibinfo{author}{\bibfnamefont{R.}~\bibnamefont{Jozsa}},
  \bibinfo{author}{\bibfnamefont{A.}~\bibnamefont{Peres}}, \bibnamefont{and}
  \bibinfo{author}{\bibfnamefont{W.~K.} \bibnamefont{Wootters}},
  \bibinfo{journal}{Physical Review Letters} \textbf{\bibinfo{volume}{70}},
  \bibinfo{pages}{1895} (\bibinfo{year}{1993}).

\bibitem[{\citenamefont{Giovannetti et~al.}(2004)\citenamefont{Giovannetti,
  Lloyd, and Maccone}}]{giovannetti2004quantum}
\bibinfo{author}{\bibfnamefont{V.}~\bibnamefont{Giovannetti}},
  \bibinfo{author}{\bibfnamefont{S.}~\bibnamefont{Lloyd}}, \bibnamefont{and}
  \bibinfo{author}{\bibfnamefont{L.}~\bibnamefont{Maccone}},
  \bibinfo{journal}{Science} \textbf{\bibinfo{volume}{306}},
  \bibinfo{pages}{1330} (\bibinfo{year}{2004}).

\bibitem[{\citenamefont{Bennett et~al.}(1984)\citenamefont{Bennett, Brassard
  et~al.}}]{bb84}
\bibinfo{author}{\bibfnamefont{C.~H.} \bibnamefont{Bennett}},
  \bibinfo{author}{\bibfnamefont{G.}~\bibnamefont{Brassard}},
  \bibnamefont{et~al.}, in \emph{\bibinfo{booktitle}{Proceedings of IEEE
  International Conference on Computers, Systems and Signal Processing}}
  (\bibinfo{organization}{New York}, \bibinfo{year}{1984}), vol.
  \bibinfo{volume}{175}, p.~\bibinfo{pages}{8}.

\bibitem[{\citenamefont{Gottesman and Chuang}(2001)}]{gottesman2001quantum}
\bibinfo{author}{\bibfnamefont{D.}~\bibnamefont{Gottesman}} \bibnamefont{and}
  \bibinfo{author}{\bibfnamefont{I.}~\bibnamefont{Chuang}},
  \bibinfo{journal}{arXiv preprint quant-ph/0105032}  (\bibinfo{year}{2001}).

\bibitem[{\citenamefont{Brassard et~al.}(2007)\citenamefont{Brassard,
  Broadbent, Fitzsimons, Gambs, and Tapp}}]{brassard2007anonymous}
\bibinfo{author}{\bibfnamefont{G.}~\bibnamefont{Brassard}},
  \bibinfo{author}{\bibfnamefont{A.}~\bibnamefont{Broadbent}},
  \bibinfo{author}{\bibfnamefont{J.}~\bibnamefont{Fitzsimons}},
  \bibinfo{author}{\bibfnamefont{S.}~\bibnamefont{Gambs}}, \bibnamefont{and}
  \bibinfo{author}{\bibfnamefont{A.}~\bibnamefont{Tapp}}, in
  \emph{\bibinfo{booktitle}{Advances in Cryptology--ASIACRYPT 2007}}
  (\bibinfo{publisher}{Springer}, \bibinfo{year}{2007}), pp.
  \bibinfo{pages}{460--473}.

\bibitem[{\citenamefont{Giovannetti et~al.}(2008)\citenamefont{Giovannetti,
  Lloyd, and Maccone}}]{giovannetti2008quantum}
\bibinfo{author}{\bibfnamefont{V.}~\bibnamefont{Giovannetti}},
  \bibinfo{author}{\bibfnamefont{S.}~\bibnamefont{Lloyd}}, \bibnamefont{and}
  \bibinfo{author}{\bibfnamefont{L.}~\bibnamefont{Maccone}},
  \bibinfo{journal}{Physical review letters} \textbf{\bibinfo{volume}{100}},
  \bibinfo{pages}{230502} (\bibinfo{year}{2008}).

\bibitem[{\citenamefont{Pironio et~al.}(2010)\citenamefont{Pironio, Ac{\'\i}n,
  Massar, de~La~Giroday, Matsukevich, Maunz, Olmschenk, Hayes, Luo, Manning
  et~al.}}]{pironio2010random}
\bibinfo{author}{\bibfnamefont{S.}~\bibnamefont{Pironio}},
  \bibinfo{author}{\bibfnamefont{A.}~\bibnamefont{Ac{\'\i}n}},
  \bibinfo{author}{\bibfnamefont{S.}~\bibnamefont{Massar}},
  \bibinfo{author}{\bibfnamefont{A.~B.} \bibnamefont{de~La~Giroday}},
  \bibinfo{author}{\bibfnamefont{D.~N.} \bibnamefont{Matsukevich}},
  \bibinfo{author}{\bibfnamefont{P.}~\bibnamefont{Maunz}},
  \bibinfo{author}{\bibfnamefont{S.}~\bibnamefont{Olmschenk}},
  \bibinfo{author}{\bibfnamefont{D.}~\bibnamefont{Hayes}},
  \bibinfo{author}{\bibfnamefont{L.}~\bibnamefont{Luo}},
  \bibinfo{author}{\bibfnamefont{T.~A.} \bibnamefont{Manning}},
  \bibnamefont{et~al.}, \bibinfo{journal}{Nature}
  \textbf{\bibinfo{volume}{464}}, \bibinfo{pages}{1021} (\bibinfo{year}{2010}).

\bibitem[{\citenamefont{Yan}(2007)}]{yan2007cryptanalytic}
\bibinfo{author}{\bibfnamefont{S.~Y.} \bibnamefont{Yan}},
  \emph{\bibinfo{title}{Cryptanalytic attacks on RSA}}
  (\bibinfo{publisher}{Springer}, \bibinfo{year}{2007}).

\bibitem[{\citenamefont{Boneh and Lipton}(1995)}]{boneh1995quantum}
\bibinfo{author}{\bibfnamefont{D.}~\bibnamefont{Boneh}} \bibnamefont{and}
  \bibinfo{author}{\bibfnamefont{R.~J.} \bibnamefont{Lipton}}, in
  \emph{\bibinfo{booktitle}{Advances in Cryptology—CRYPT0’95}}
  (\bibinfo{publisher}{Springer}, \bibinfo{year}{1995}), pp.
  \bibinfo{pages}{424--437}.

\bibitem[{\citenamefont{Brassard et~al.}(1998)\citenamefont{Brassard, H{\o}yer,
  and Tapp}}]{brassard1998quantum}
\bibinfo{author}{\bibfnamefont{G.}~\bibnamefont{Brassard}},
  \bibinfo{author}{\bibfnamefont{P.}~\bibnamefont{H{\o}yer}}, \bibnamefont{and}
  \bibinfo{author}{\bibfnamefont{A.}~\bibnamefont{Tapp}}, in
  \emph{\bibinfo{booktitle}{LATIN'98: Theoretical Informatics}}
  (\bibinfo{publisher}{Springer}, \bibinfo{year}{1998}), pp.
  \bibinfo{pages}{163--169}.

\bibitem[{\citenamefont{Mayers}(1997)}]{mayers1997unconditionally}
\bibinfo{author}{\bibfnamefont{D.}~\bibnamefont{Mayers}},
  \bibinfo{journal}{Physical review letters} \textbf{\bibinfo{volume}{78}},
  \bibinfo{pages}{3414} (\bibinfo{year}{1997}).

\bibitem[{\citenamefont{Lo}(1997)}]{lo1997insecurity}
\bibinfo{author}{\bibfnamefont{H.-K.} \bibnamefont{Lo}},
  \bibinfo{journal}{Physical Review A} \textbf{\bibinfo{volume}{56}},
  \bibinfo{pages}{1154} (\bibinfo{year}{1997}).

\bibitem[{\citenamefont{Gentry}(2009)}]{Gentry09}
\bibinfo{author}{\bibfnamefont{C.}~\bibnamefont{Gentry}},
  \bibinfo{journal}{Proceedings of the 41st annual ACM Symposium on Theory of
  Computing (STOC)} pp. \bibinfo{pages}{169--178} (\bibinfo{year}{2009}).

\bibitem[{\citenamefont{van Dijk et~al.}(2010)\citenamefont{van Dijk, Gentry,
  Halevi, and Vaikuntanathan}}]{Dijk09}
\bibinfo{author}{\bibfnamefont{M.}~\bibnamefont{van Dijk}},
  \bibinfo{author}{\bibfnamefont{C.}~\bibnamefont{Gentry}},
  \bibinfo{author}{\bibfnamefont{S.}~\bibnamefont{Halevi}}, \bibnamefont{and}
  \bibinfo{author}{\bibfnamefont{V.}~\bibnamefont{Vaikuntanathan}}, in
  \emph{\bibinfo{booktitle}{Advances in Cryptology---EUROCRYPT 2010, Lecture
  Notes in Computer Science}} (\bibinfo{year}{2010}), vol.
  \bibinfo{volume}{6110}, pp. \bibinfo{pages}{24--43}.

\bibitem[{\citenamefont{Smart and Vercauteren}(2010)}]{smart2010fully}
\bibinfo{author}{\bibfnamefont{N.~P.} \bibnamefont{Smart}} \bibnamefont{and}
  \bibinfo{author}{\bibfnamefont{F.}~\bibnamefont{Vercauteren}}, in
  \emph{\bibinfo{booktitle}{Public Key Cryptography--PKC 2010}}
  (\bibinfo{publisher}{Springer}, \bibinfo{year}{2010}), pp.
  \bibinfo{pages}{420--443}.

\bibitem[{\citenamefont{Brakerski and
  Vaikuntanathan}(2011)}]{brakerski2011efficient}
\bibinfo{author}{\bibfnamefont{Z.}~\bibnamefont{Brakerski}} \bibnamefont{and}
  \bibinfo{author}{\bibfnamefont{V.}~\bibnamefont{Vaikuntanathan}}, in
  \emph{\bibinfo{booktitle}{Foundations of Computer Science (FOCS), 2011 IEEE
  52nd Annual Symposium on}} (\bibinfo{organization}{IEEE},
  \bibinfo{year}{2011}), pp. \bibinfo{pages}{97--106}.

\bibitem[{\citenamefont{Rivest et~al.}(1978)\citenamefont{Rivest, Adleman, and
  Dertouzos}}]{rivest1978data}
\bibinfo{author}{\bibfnamefont{R.~L.} \bibnamefont{Rivest}},
  \bibinfo{author}{\bibfnamefont{L.}~\bibnamefont{Adleman}}, \bibnamefont{and}
  \bibinfo{author}{\bibfnamefont{M.~L.} \bibnamefont{Dertouzos}},
  \bibinfo{journal}{Foundations of secure computation}
  \textbf{\bibinfo{volume}{4}}, \bibinfo{pages}{169} (\bibinfo{year}{1978}).

\bibitem[{\citenamefont{Broadbent et~al.}(2009)\citenamefont{Broadbent,
  Fitzsimons, and Kashefi}}]{bfk09}
\bibinfo{author}{\bibfnamefont{A.}~\bibnamefont{Broadbent}},
  \bibinfo{author}{\bibfnamefont{J.}~\bibnamefont{Fitzsimons}},
  \bibnamefont{and} \bibinfo{author}{\bibfnamefont{E.}~\bibnamefont{Kashefi}},
  \bibinfo{journal}{Proceedings of the 50th Annual IEEE Symposium on
  Foundations of Computer Science (FOCS 2009)} pp. \bibinfo{pages}{517--526}
  (\bibinfo{year}{2009}).

\bibitem[{\citenamefont{Fitzsimons and
  Kashefi}(2012)}]{fitzsimons2012unconditionally}
\bibinfo{author}{\bibfnamefont{J.~F.} \bibnamefont{Fitzsimons}}
  \bibnamefont{and} \bibinfo{author}{\bibfnamefont{E.}~\bibnamefont{Kashefi}},
  \bibinfo{journal}{arXiv preprint arXiv:1203.5217}  (\bibinfo{year}{2012}).

\bibitem[{\citenamefont{Morimae and Fujii}(2013)}]{morimae2013blind}
\bibinfo{author}{\bibfnamefont{T.}~\bibnamefont{Morimae}} \bibnamefont{and}
  \bibinfo{author}{\bibfnamefont{K.}~\bibnamefont{Fujii}},
  \bibinfo{journal}{Physical Review A} \textbf{\bibinfo{volume}{87}},
  \bibinfo{pages}{050301} (\bibinfo{year}{2013}).

\bibitem[{\citenamefont{Mantri et~al.}(2013)\citenamefont{Mantri,
  P{\'e}rez-Delgado, and Fitzsimons}}]{mantri2013optimal}
\bibinfo{author}{\bibfnamefont{A.}~\bibnamefont{Mantri}},
  \bibinfo{author}{\bibfnamefont{C.~A.} \bibnamefont{P{\'e}rez-Delgado}},
  \bibnamefont{and} \bibinfo{author}{\bibfnamefont{J.~F.}
  \bibnamefont{Fitzsimons}}, \bibinfo{journal}{Physical Review Letters}
  \textbf{\bibinfo{volume}{111}}, \bibinfo{pages}{230502}
  (\bibinfo{year}{2013}).

\bibitem[{\citenamefont{Barz et~al.}(2012)\citenamefont{Barz, Kashefi,
  Broadbent, Fitzsimons, Zeilinger, and Walther}}]{barz2012demonstration}
\bibinfo{author}{\bibfnamefont{S.}~\bibnamefont{Barz}},
  \bibinfo{author}{\bibfnamefont{E.}~\bibnamefont{Kashefi}},
  \bibinfo{author}{\bibfnamefont{A.}~\bibnamefont{Broadbent}},
  \bibinfo{author}{\bibfnamefont{J.~F.} \bibnamefont{Fitzsimons}},
  \bibinfo{author}{\bibfnamefont{A.}~\bibnamefont{Zeilinger}},
  \bibnamefont{and} \bibinfo{author}{\bibfnamefont{P.}~\bibnamefont{Walther}},
  \bibinfo{journal}{Science} \textbf{\bibinfo{volume}{335}},
  \bibinfo{pages}{303} (\bibinfo{year}{2012}).

\bibitem[{\citenamefont{Barz et~al.}(2013)\citenamefont{Barz, Fitzsimons,
  Kashefi, and Walther}}]{barz2013experimental}
\bibinfo{author}{\bibfnamefont{S.}~\bibnamefont{Barz}},
  \bibinfo{author}{\bibfnamefont{J.~F.} \bibnamefont{Fitzsimons}},
  \bibinfo{author}{\bibfnamefont{E.}~\bibnamefont{Kashefi}}, \bibnamefont{and}
  \bibinfo{author}{\bibfnamefont{P.}~\bibnamefont{Walther}},
  \bibinfo{journal}{Nature Physics}  (\bibinfo{year}{2013}).

\bibitem[{\citenamefont{Liang}(2013)}]{MinL13}
\bibinfo{author}{\bibfnamefont{M.}~\bibnamefont{Liang}},
  \bibinfo{journal}{Quantum Inf. Process.} \textbf{\bibinfo{volume}{12}},
  \bibinfo{pages}{3675} (\bibinfo{year}{2013}).

\bibitem[{\citenamefont{Childs}(2005)}]{Ch05}
\bibinfo{author}{\bibfnamefont{A.}~\bibnamefont{Childs}},
  \bibinfo{journal}{Quantum Information and Computation}
  \textbf{\bibinfo{volume}{5}}, \bibinfo{pages}{456} (\bibinfo{year}{2005}).

\bibitem[{\citenamefont{Fisher et~al.}(2014)\citenamefont{Fisher, Broadbent,
  Shalm, Yan, Lavoie, Prevedel, Jennewein, and Resch}}]{Fisher13}
\bibinfo{author}{\bibfnamefont{K.}~\bibnamefont{Fisher}},
  \bibinfo{author}{\bibfnamefont{A.}~\bibnamefont{Broadbent}},
  \bibinfo{author}{\bibfnamefont{L.}~\bibnamefont{Shalm}},
  \bibinfo{author}{\bibfnamefont{Z.}~\bibnamefont{Yan}},
  \bibinfo{author}{\bibfnamefont{J.}~\bibnamefont{Lavoie}},
  \bibinfo{author}{\bibfnamefont{R.}~\bibnamefont{Prevedel}},
  \bibinfo{author}{\bibfnamefont{T.}~\bibnamefont{Jennewein}},
  \bibnamefont{and} \bibinfo{author}{\bibfnamefont{K.}~\bibnamefont{Resch}},
  \bibinfo{journal}{Nat. Commun.} \textbf{\bibinfo{volume}{5}},
  \bibinfo{pages}{3074} (\bibinfo{year}{2014}).

\bibitem[{\citenamefont{Rohde et~al.}(2012)\citenamefont{Rohde, Fitzsimons, and
  Gilchrist}}]{rfg12}
\bibinfo{author}{\bibfnamefont{P.~P.} \bibnamefont{Rohde}},
  \bibinfo{author}{\bibfnamefont{J.~F.} \bibnamefont{Fitzsimons}},
  \bibnamefont{and}
  \bibinfo{author}{\bibfnamefont{A.}~\bibnamefont{Gilchrist}},
  \bibinfo{journal}{Phys. Rev. Lett.} \textbf{\bibinfo{volume}{109}},
  \bibinfo{pages}{150501} (\bibinfo{year}{2012}).

\bibitem[{\citenamefont{Griffiths}(2005)}]{Grif05}
\bibinfo{author}{\bibfnamefont{R.~B.} \bibnamefont{Griffiths}},
  \bibinfo{journal}{Phys. Rev. A} \textbf{\bibinfo{volume}{71}},
  \bibinfo{pages}{042337} (\bibinfo{year}{2005}).

\bibitem[{\citenamefont{Nielsen and Chuang}(1997)}]{nc97}
\bibinfo{author}{\bibfnamefont{M.~A.} \bibnamefont{Nielsen}} \bibnamefont{and}
  \bibinfo{author}{\bibfnamefont{I.~L.} \bibnamefont{Chuang}},
  \bibinfo{journal}{Phys. Rev. Lett.} \textbf{\bibinfo{volume}{79}},
  \bibinfo{pages}{321} (\bibinfo{year}{1997}).

\bibitem[{\citenamefont{Broadbent et~al.}(2013)\citenamefont{Broadbent,
  Gutoski, and Stebila}}]{broadbent2013quantum}
\bibinfo{author}{\bibfnamefont{A.}~\bibnamefont{Broadbent}},
  \bibinfo{author}{\bibfnamefont{G.}~\bibnamefont{Gutoski}}, \bibnamefont{and}
  \bibinfo{author}{\bibfnamefont{D.}~\bibnamefont{Stebila}}, in
  \emph{\bibinfo{booktitle}{Advances in Cryptology--CRYPTO 2013}}
  (\bibinfo{publisher}{Springer}, \bibinfo{year}{2013}), pp.
  \bibinfo{pages}{344--360}.

\end{thebibliography}

\end{document}